\documentclass[prl,
superscriptaddress,
longbibliography,
twocolumn
]{revtex4-1}
\usepackage[latin1]{inputenc}
\usepackage{graphicx}
\usepackage{amsmath}
\usepackage{amsthm}
\usepackage{bm}
\usepackage{layout}
\usepackage{float}
\usepackage{amsfonts}
\usepackage{amssymb}%
\usepackage[margin=0.75in]{geometry}
\usepackage{color}
\usepackage[colorlinks=true,citecolor=blue]{hyperref}
\usepackage{soul}
\usepackage{mathtools}

\newtheorem{lemma}{Lemma}

\usepackage[capitalise,compress]{cleveref}

\renewcommand{\epsilon}{\varepsilon}
\renewcommand{\O}[1]{O\left(#1\right)}
\newcommand{\BigTheta}[1]{\Theta\left(#1\right)}

\newcommand{\norm}[1]{\left\|#1\right\|}

\newcommand{\comm}[1]{\left[#1\right]}

\newcounter{para}

\makeatletter
\newcommand*\bigcdot{\mathpalette\bigcdot@{.5}}
\newcommand*\bigcdot@[2]{\mathbin{\vcenter{\hbox{\scalebox{#2}{$\m@th#1\bullet$}}}}}

\makeatother

\usepackage{braket}

\usepackage{array}
\newcolumntype{L}{>{$}l<{$}} 
\newcolumntype{C}{>{$}c<{$}} 
\newcolumntype{R}{>{$}r<{$}} 

\usepackage{xcolor}

\usepackage{mathtools}
\DeclarePairedDelimiter\ceil{\lceil}{\rceil}

\newcommand{\Norm}[1]{\left\lVert#1\right\rVert}

\usepackage{xr}
\makeatletter
\newcommand*{\addFileDependency}[1]{
  \typeout{(#1)}
  \@addtofilelist{#1}
  \IfFileExists{#1}{}{\typeout{No file #1.}}
}
\makeatother

\renewcommand{\o}{1}
\newcommand{\e}{2}
\newcommand{\J}{J}

\newcommand{\etal}{et al.}

\newcommand{\Section}[1]{\textit{#1.---}}
\begin{document}

\title{Destructive Error Interference in Product-Formula Lattice Simulation}

\author{Minh~C.~Tran}
\affiliation{Joint Center for Quantum Information and Computer Science, NIST/University of Maryland, College Park, Maryland 20742, USA}
\affiliation{Joint Quantum Institute, NIST/University of Maryland, College Park, Maryland 20742, USA}
\affiliation{Kavli Institute for Theoretical Physics, University of California, Santa Barbara, California 93106, USA}
\author{Su-Kuan Chu}
\affiliation{Joint Center for Quantum Information and Computer Science, NIST/University of Maryland, College Park, Maryland 20742, USA}
\affiliation{Joint Quantum Institute, NIST/University of Maryland, College Park, Maryland 20742, USA}
\author{Yuan~Su}
\affiliation{Joint Center for Quantum Information and Computer Science, NIST/University of Maryland, College Park, Maryland 20742, USA}
\affiliation{Department of Computer Science, University of Maryland, College Park, Maryland 20742, USA}
\affiliation{Institute for Advanced Computer Studies, University of Maryland, College Park, Maryland 20742, USA}
\author{Andrew~M.~Childs}
\affiliation{Joint Center for Quantum Information and Computer Science, NIST/University of Maryland, College Park, Maryland 20742, USA}
\affiliation{Department of Computer Science, University of Maryland, College Park, Maryland 20742, USA}
\affiliation{Institute for Advanced Computer Studies, University of Maryland, College Park, Maryland 20742, USA}
\author{Alexey~V.~Gorshkov}
\affiliation{Joint Center for Quantum Information and Computer Science, NIST/University of Maryland, College Park, Maryland 20742, USA}
\affiliation{Joint Quantum Institute, NIST/University of Maryland, College Park, Maryland 20742, USA}
\affiliation{Kavli Institute for Theoretical Physics, University of California, Santa Barbara, California 93106, USA}

\begin{abstract}
Quantum computers can efficiently simulate the dynamics of quantum systems.
In this paper, we study the cost of digitally simulating the dynamics of several physically relevant systems using the first-order product formula algorithm.
We show that the errors from different Trotterization steps in the algorithm can interfere destructively, yielding a much smaller error than previously estimated.
In particular, we prove that the total error in simulating a nearest-neighbor interacting system of $n$ sites for time $t$ using the first-order product formula with $r$ time slices is $\O{{nt}/{r}+{nt^3}/{r^2}}$ when $nt^2/r$ is less than a small constant.
Given an error tolerance $\epsilon$, the error bound yields an estimate of $\max\{O({n^2t}/{\epsilon}),O({n^2 t^{3/2}}/{\epsilon^{1/2}})\}$ for the total gate count of the simulation. 
The estimate is tighter than previous bounds and matches the empirical performance observed in Childs et al.~[PNAS 115, 9456-9461 (2018)].
We also provide numerical evidence for potential improvements and conjecture an even tighter estimate for the gate count. 
\end{abstract}
\maketitle
Simulating the dynamics of quantum systems is one of the primary applications of quantum computers. 
While analog quantum simulations rely on engineering physical systems to mimic other systems, 
digital quantum simulations use algorithms to decompose the evolution unitary into a sequence of elementary quantum gates.
The first quantum simulation algorithm proposed by Lloyd~\cite{Lloyd1996} uses  the Lie-Trotter product formula, also known as the first-order product formula~(PF1)~\cite{Suzuki1985,Huyghebaert1990}. 
Since then, more advanced quantum simulation algorithms have been developed, including algorithms based on the higher-order product formulae~\cite{ChildsY19,Childs2004, BerryACS07,ChildsMNRS2017}, linear combinations of unitaries~\cite{BerryCCKS2015,Low19}, quantum signal processing~\cite{LowC2017}, and Lieb-Robinson bounds~\cite{Haah,Tran2018}, which all asymptotically reduce the cost of digital quantum simulation in terms of the number of gates used in the limit of large time or large system size.

Despite these developments, PF1 remains one of the most popular algorithms for near-term implementations of digital quantum simulation due to its simplicity.
In practice, the small prefactor in the scaling of the gate count of PF1 compared to more advanced quantum simulation algorithms makes it attractive for simulations where the evolution time and the system size are not too large~\cite{ChildsMNRS2017}. 

Despite its simplicity and wide applicability, a tight error bound for PF1 in simulating many physically relevant systems remains elusive. 
Recent works~\cite{Suzuki1985,Huyghebaert1990,ChildsY19} estimated that $\O{n^2t^2}$ elementary gates suffice to simulate the dynamics of a nearest-neighbor interacting system consisting of $n$ sites for time~$t$ using PF1~\footnote{Refs.~\cite{Suzuki1985,Huyghebaert1990,ChildsY19} took into account the commutativity between some interaction terms in the Hamiltonian of a nearest-neighbor interacting system. 
Without this commutativity, the gate count would be $\O{n^3t^2}$~\cite{Childs2004,BerryACS07}}.
However, the numerical evidence in Ref.~\cite{ChildsMNRS2017} suggests that PF1 performs much better than this in practice. 
In particular, the gate count for simulating the dynamics of a nearest-neighbor Heisenberg spin chain of length $n$ for time $t=n$ scales only as $\O{n^{2.964}}$.  
In addition, Heyl~\etal~\cite{Heyl18} also found that the error of simulating the time evolution of a local observable using PF1 can be much smaller than theoretically estimated.

In this paper, we provide an approach to tighten the error bound of PF1 for simulating several physically relevant systems, including those with nearest-neighbor interactions.
The key finding of the paper is that the errors from different steps of the algorithm can combine destructively, resulting in a smaller total error than previous analysis estimates.
In particular, the tighter error bound suggests that simulating the dynamics of a nearest-neighbor interacting system of $n$ sites for time $t$ up to an error tolerance $\epsilon$ requires only $\max\left\{\O{n^2 t/\epsilon},\O{n^2 t^{3/2}/\epsilon^{1/2}}\right\}$ quantum gates, which is asymptotically smaller than
the state-of-the-art bound $\O{n^2 t^2/\epsilon}$ in Refs.~\cite{Suzuki1985,Huyghebaert1990,ChildsY19}. 
At $t=n$ and at a fixed $\epsilon$, our estimate $\O{n^3}$ also closely matches the empirical gate count $\O{n^{2.964}}$ computed in Ref.~\cite{ChildsMNRS2017}.

\Section{Setup}
We assume that the system evolves under a Hamiltonian $H=\sum_X h_X$, which is a sum of time-independent terms $h_X$, each acting nontrivially on a subset $X$ of constant size.
Our approach applies if there exists a partition $H=H_1+H_2$ such that the terms $h_X$ in $H_1$ mutually commute and the terms $h_X$ in $H_2$ also mutually commute.
Examples of Hamiltonians that satisfy this assumption include all one-dimensional, finite-range~\footnote{
For interactions of maximum range 
$R$, we can group $\ceil{ (R+1)/2 }$ 
consecutive sites into distinct blocks such that the Hamiltonian consists of only interactions between nearest-neighbor blocks. 
The error analysis for using the first-order product formula to simulate such a system would follow from our analysis for simulating nearest-neighbor interactions.
Note, however, that we assume that the exact simulation of the evolution of each constant-size block requires only a constant amount of elementary gates. 
} interacting systems, such as the Heisenberg model and the transverse field Ising model in one dimension with either open or periodic boundary conditions, and with or without disorder.
Additionally, this assumption also covers some physically relevant systems in higher dimensions, such as the transverse field Ising model with either finite-range or long-range interactions.

To simulate the time-evolution of the system for time $t$ using elementary quantum gates, we use the first-order product formula~\cite{Lloyd1996}:
\begin{align}
U_{t} \approx \left[U_{t/r}^{(\o)}U_{t/r}^{(\e)}\right]^{r} ,\label{EQ_Delta_def}
\end{align}
where $U_t\coloneqq \exp(-i H t)$, $U^{(p)}_{t/r} \coloneqq \exp(-i H_{p} t/r)$ for $p=1,2$, and $r$ is the number of time segments to be chosen later so that the norm of the total error ${\Delta} \coloneqq {U_{t}-[U_{t/r}^{(\o)}U_{t/r}^{(\e)}]^{r}}$ is at most a constant $\epsilon$. 
By our assumption that the terms within $H_p$ ($p=1,2$) mutually commute, we can further decompose the evolution $U^{(p)}_{t/r}$ into a product of elementary quantum gates with no additional error.

For simplicity, we demonstrate our approach to estimating the gate count of PF1 on a one-dimensional lattice of $n$ sites, evolving under a time-independent, nearest-neighbor Hamiltonian $H=\sum_{i=1}^{n-1}h_{i},$ where $h_i$ is supported only on sites $i,i+1$, $\Norm{h_i}\leq \J$ for all $i$,
$\J$ is a constant, and $\norm{\cdot}$ denotes the operator norm.
Without loss of generality, we also assume $\J =1$, which sets the time scale for the dynamics of the system. 
We then apply PF1 to the partition $H=H_{\o}+H_{\e}$, where $H_\o = \sum_{\text{odd } j} h_j$ and $H_\e = \sum_{\text{even }j} h_j$.
Note that the terms within $H_\o$ ($H_\e$) mutually commute and therefore satisfy the aforementioned assumption.

\Section{Leading contributions}
To estimate the gate count, we first need a bound on the total error $\Delta$.
The previous best bound from Ref.~\cite{ChildsY19} gives $\norm{\Delta}\leq \O{nt^2/r}$, so that $r = \BigTheta{nt^2/\epsilon}$ suffices to ensure the total error at most $\epsilon$, giving gate count $nr = \O{t^2n^2/\epsilon}$.
Before we prove our tighter bound, we will first argue simply based on the lowest order error that $\norm{\Delta}\approx \O{nt/r}$, which would result in a gate count $\O{tn^2/\epsilon}$, matching the empirical estimate of about $\O{n^3}$ for $t=n$ in Ref.~\cite{ChildsMNRS2017}.

Let $\delta = U_{t/r}-U^{(\o)}_{t/r}U^{(\e)}_{t/r}$ be the error of the approximation in each time segment. 
In the limit $r\gg t$, the leading contribution to $\delta$ is given by the commutator between $H_\o$ and $H_\e$~\cite{Lloyd1996}:
\begin{align}
	\norm{\delta}
	\approx\frac{1}{2}\frac{t^2}{r^2}\norm{\comm{H_\o,H_\e}}= \O{\frac{nt^2}{r^2}}.\label{eq:delta_first_order}
\end{align}
Replacing $U^{(\o)}_{t/r}U^{(\e)}_{t/r}$ by $U_{t/r}+\delta$ on the right-hand side of \cref{EQ_Delta_def} and expanding to first order in $\delta$, we have an approximation for the total error:
\begin{align}
\Delta &\approx \sum_{j=0}^{r-1}U_{t/r}^j~\delta~ U_{t/r}^{r-1-j}
=\left(\sum_{j=0}^{r-1}U_{t/r}^j~\delta~ U_{t/r}^{-j}\right)U_{t/r}^{r-1},\label{EQ_Delta_sum}
\end{align}
where $U_{t/r}^j \coloneqq (U_{t/r})^j$. 
If we bound $\norm{\Delta}$ using the triangle inequality, i.e.,
\begin{align}
\norm{\Delta}
&\approx\Norm{\sum_{j=0}^{r-1}U_{t/r}^j~\delta~ U_{t/r}^{r-1-j}}
\leq r \norm{\delta}
\approx \O{\frac{nt^2}{r}},\label{eq:trianglebound}
\end{align}
we get the same error bound (and hence the same gate count) as Ref.~\cite{ChildsY19}.

To understand the key idea for improving the bound, imagine the unitary evolution $U_{t/r}^j\delta U_{t/r}^{-j}$ as a rotation of $\delta$ by a small angle proportional to $jt/r$. 
\Cref{EQ_Delta_sum} sums over the rotations of $\delta$ by evenly spaced angles.
Therefore, the sum involves significant cancellation, making it much smaller than the upper bound derived using the triangle inequality [\cref{eq:trianglebound}]. 

To realize this intuition, we make a change of variables to $x = tj/r$ and approximate the sum in $\Delta$ by an integral:
\begin{align}
\norm{\Delta} &\approx \Norm{\sum_{j=0}^{r-1}U_{t/r}^j~\delta~ U_{t/r}^{-j}}
\approx\frac{r}{t}\Norm{\int_0^t dx U_{x} ~\delta~ U_{-x}}.\label{EQ_delta_int}
\end{align}
With the assumption that $H = H_1 + H_2$ is a sum of two terms, we rewrite $\delta$ (to leading order in $t/r$) as
\begin{align}
\delta\approx\frac{1}{2}\comm{H_\o,H_\e}\frac{t^2}{r^2} = \frac{1}{2}\comm{H,H_\e}\frac{t^2}{r^2},\label{EQ_delta_comm}
\end{align}
and use the identity
\begin{align}
U_t ~A~U_{-t} - A = -i \int_{0}^t dx U_x\comm{H,A} U_{-x},\label{EQ_int_id}
\end{align}
with $A = \frac{t^2}{2r^2} H_\e $, to evaluate the integral in \cref{EQ_delta_int} and arrive at an estimate for the norm of $\Delta$:
\begin{align}
\norm{\Delta}
&\approx\frac{r}{t}\Norm{\int_0^t dx U_{x} ~\comm{H,\frac{t^2}{2r^2}H_\e}~ U_{-x}}\nonumber\\
&
\leq \frac{t}{2r} 2\norm{H_\e} = \O{\frac{nt}{r}},\label{eq:rough_error_bound}
\end{align}
which is a factor of $t$ tighter than \cref{eq:trianglebound}.
To ensure that the total error $\norm{\Delta}$ is at most $\epsilon$, we choose $r =  \BigTheta {nt/\epsilon}$, leading to the total gate count $\O{nr}= \O {n^2 t/\epsilon}$, which has optimal scaling in $t$~\cite{BerryCK15}.
At $t=n$ and fixed $\epsilon$, the gate count becomes $\O{n^3}$, which closely matches the empirical performance $\O{n^{2.964}}$ observed in Ref.~\cite{ChildsMNRS2017}.

Additionally, if the time step $t/r=\tau$ is a constant, the total error of the simulation $\norm{\Delta} = \O{n\tau}$ appears to be independent of the total number of time segments. 
This feature agrees well with Ref.~\cite{Heyl18}, where the authors argue that for a fixed, small value of $\tau$, the error in simulating the evolution of a local observable using PF1 would not increase with the total simulation time~$t$.    
However, our bound is more general; it applies to the error in simulating the evolution unitary of the system, and hence any observable.

\Section{Higher-order contributions}
We made three approximations in deriving \cref{eq:rough_error_bound}.
First, in \cref{EQ_delta_comm}, we considered $\delta$ to only the leading order in $t/r$ and discarded terms of higher order in $t/r$.
We then expanded $\Delta$ in \cref{EQ_Delta_sum} to only the first order in $\delta$ while ignoring the higher-order terms in $\delta^{k}$. 
Additionally, we evaluated the sum in \cref{EQ_delta_int} by approximating it with an integral.
We now make the estimation rigorous by considering the errors incurred upon making the three approximations.

First, we show that higher-order terms in $t/r$ in the expansion of $\delta$ are indeed dominated by the second order.
For that, we write $\delta$ as a series in $t/r$:
\begin{equation}
\delta  \coloneqq U_{t/r}-U^{(\o)}_{t/r}U^{(\e)}_{t/r}  
=\sum_{k=2}^{\infty}\frac{(-it)^{k}}{k!r^{k}}\delta_{k},\label{EQ_delta_expand}
\end{equation}
where $\delta_k$ are operators independent of $t,r$.
If we only need a bound on the norm of $\delta$, it is sufficient to bound the norms of $\delta_k$.
However, in addition to the norm, we are also interested in the structure of $\delta_k$, described in \cref{TH_delta_k}, which is crucial for evaluating the total error [See \cref{EQ_delta_comm}].

\begin{lemma}
\label{TH_delta_k}
For all $k\geq 2$, there exist $S_k, V_k$ such that $\delta_k = \comm{H,S_k}+V_k$ and
\begin{align}
&\norm{V_k} = \O{e^{k-2} n^{k-2}},\label{EQ_Vk_norm}\\
&\norm{S_k} = \O{k^2 n^{k-1}},\label{EQ_Sk_norm}\\
&\norm{\comm{H,S_k}}  = \O{k^3 n^{k-1}},\label{EQ_HSk_norm}
\end{align}
where the big-$O$ constants do not depend on $k$.
\end{lemma}
\Cref{TH_delta_k} holds for $k = 2$, with $S_2 = H_\e$ and $V_2 = 0$ [See \cref{EQ_delta_comm}].
For $k>2$, we construct $S_k,V_k$ inductively using the definition of $\delta_k$ in \cref{EQ_delta_expand}.
The factor $n^{k-2}$ in the norm of $V_k$ comes from the $(k-2)$-th nested commutators in the expansion of $\delta_k$. 
We provide a detailed proof of the lemma in the Supplemental Material (SM)~\cite{SM}.

A corollary of \cref{TH_delta_k} is $\norm{\delta_k} = \O{e^k n^{k-1}}$, and therefore, we can immediately bound the norm of $\delta$:
\begin{align}
	\norm{\delta}
	&\leq \sum_{k=2}^{\infty} \frac{t^k}{k! r^k}\norm{\delta_k} 
	= \O{\frac{n t^2}{r^2}\sum_{k=0}^\infty \frac{(ent)^k}{k! r^k}}\nonumber\\
	&=\O{\frac{nt^2}{r^2}\exp{\frac{e nt}{r}}}=\O{\frac{n t^2}{r^2}},\label{EQ_norm_delta}
\end{align}
where we assume $r>ent$. We later fulfill this condition by choosing an appropriate value for $r$.

Another corollary of \Cref{TH_delta_k} is that $\delta = \comm{H,S} + V$, where
$
	S = \sum_{k=2}^\infty \frac{(-it)^{k}}{k!r^{k}} S_k$ and $ V = \sum_{k=3}^\infty \frac{(-it)^{k}}{k!r^{k}} V_k.
$
It is straightforward to verify the bounds on the norms of $S$ and $V$:
\begin{equation}
	\norm{S} = \O{\frac{nt^2}{r^2}},\quad
	\norm{V} = \O{\frac{n t^3}{r^3}},\label{EQ_norm_SV}
\end{equation}
where we again assume $r>ent$.

Next, we rectify the approximation in \cref{EQ_delta_int} by rigorously bounding the norm of the sum. 
\begin{lemma}
\label{LEM_sum_delta}
For any positive integer $a\geq 1$,
\begin{align}
	\Norm{\sum_{j=0}^{a-1} U^{j}_{t/r} \delta \: U^{-j}_{t/r} } = \O{\frac{nt}{r}}+\O{a\frac{nt^3}{r^3}}. \label{eq:sumboundmain}
\end{align}
\end{lemma}
When $a = r$, the left-hand side of \cref{eq:sumboundmain} is exactly the sum in \cref{EQ_delta_int}.
We bound the sum by approximating it with an integral, which yields $\O{nt/r}$ after evaluation.
Carefully bounding the error of the approximation results in the second term $\O{ant^3/r^3}$.
We present the detailed proof of the lemma in the SM~\cite{SM}.

Given \cref{TH_delta_k} and \cref{LEM_sum_delta}, we now bound the total error $\norm\Delta$. 
We expand $\Delta$ as a series in $\delta$ and write 
$
\Delta=\sum_{k=1}^{r}\Delta_{k},
$
where $\Delta_{k}$ involves only the $k$-th order in $\delta$. 
For example, $\Delta_{1} = 	\sum_{j=0}^{r-1} U_{t/r}^{j} \delta \: U_{t/r}^{-j}$, the norm of which we can already bound using \cref{LEM_sum_delta}. 
We can use the same technique to estimate $\norm {\Delta_k}$ for all $k\geq 1$~\cite{SM}: 
\begin{align}
\norm{ \Delta_{k}} & \leq r^{k-1}\left\Vert \delta\right\Vert ^{k-1}\O{\frac{nt}{r}+\frac{nt^3}{r^2}}.\label{LEM_Delta_k_norm}
\end{align}
Finally, we bound $\norm\Delta$ using the triangle inequality:

\begin{align}
\norm{\Delta} &\leq  \sum_{k=1}^{r}\norm{\Delta_{k}}  
=\O{\frac{nt}{r}+\frac{nt^3}{r^2}},\label{EQ_norm_Delta}
\end{align}
where we assume $r\left\Vert \delta\right\Vert <1/2$ so that $\sum_{k=1}^r (r\norm{\delta})^{k-1} = \O{1}$. 
With our choice of $r$, this assumption later reduces to $\epsilon t\leq 1$, where $\epsilon$ is the error tolerance of the simulation.

\begin{figure}[t]
\centering
\includegraphics[width=0.45\textwidth]{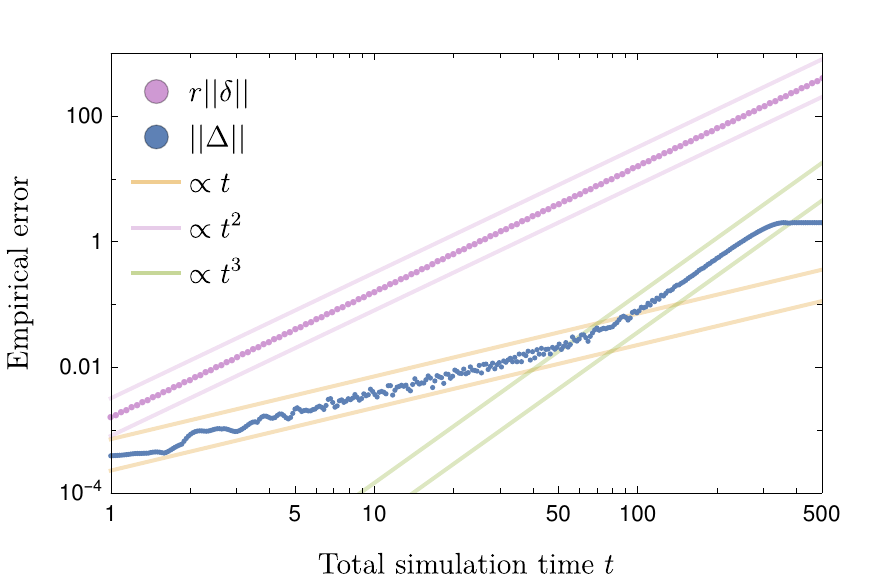}
\caption{The total error $\norm\Delta$ (blue dots) of PF1 in simulating the Heisenberg chain in \cref{EQ_Heisenberg} is numerically evaluated at $n=8$, $r=10000$, and variable time $t$ between 0 and 1000.
The purple dots represent the error estimate $r\norm{\delta}$ one would get using the triangle inequalities [\cref{eq:trianglebound}]. 
We also plot functions proportional to $t$ (orange lines), $t^2$ (purple lines), and $t^3$ (green lines) for reference.
}
\label{fig:error_scaling}
\end{figure}

\Section{Empirical error scaling}
We now benchmark the bound in \cref{EQ_norm_Delta} against the empirical error in simulating the dynamics of a nearest-neighbor Heisenberg chain:
\begin{equation}
H=\sum_{i=1}^{n-1}\vec \sigma_i \cdot \vec \sigma_{i+1}, \label{EQ_Heisenberg}
\end{equation}
where $\vec \sigma_i = (\sigma^x_i,\sigma^y_i,\sigma^z_i)$ denotes the Pauli matrices on qubit $i$. 
Using fixed values for $n$ and $r$, we compute the total error of PF1 at different times $t$ and plot the result in \cref{fig:error_scaling}.
We also plot in \cref{fig:error_scaling_higher} the empirical errors of simulating the same system using the second-order (PF2) and the fourth-order (PF4) product formulae~\cite{BerryACS07}. 

From \cref{fig:error_scaling}, the total error of PF1 appears to agree well with our bound in \cref{EQ_norm_Delta}. 
The change in the error scaling from $O(t)$ at small time to $O(t^3)$ at large time can be explained by the destructive error interference between the time slices as follows. 
While the leading error terms in each time slice scale as $O(t^2)$, they interfere destructively between time slices, resulting in a total contribution that increases with time at a slower rate $O(t)$ [recall \cref{eq:rough_error_bound}].
Meanwhile, some higher-order error terms do not interfere destructively. 
They scale as $O(t^3)$ and eventually take over as the primary contribution to the total error.
This intuition also explains the similarity between the error scalings of PF1 (at late time) and PF2 [\cref{fig:error_scaling_higher}]. 
On the other hand, if there were no destructive error interference between the time slices, the contribution from the leading error terms to the total error of PF1 would have scaled as $O(t^2)$ [\cref{fig:error_scaling}, purple dots] and saturated at 2 before the higher-order terms could take over.

We also note that the error of PF2 [PF4] scales as $t^3$ [$t^5$] initially before saturating at a later time, in agreement with the existing bounds using triangle inequalities for the higher-order product formulae~\cite{ChildsMNRS2017,ChildsY19}.   
Therefore, the destructive interference of the errors between the time segments appears to be a unique feature of the first-order product formula.

\begin{figure}[t]
\centering
\includegraphics[width=0.45\textwidth]{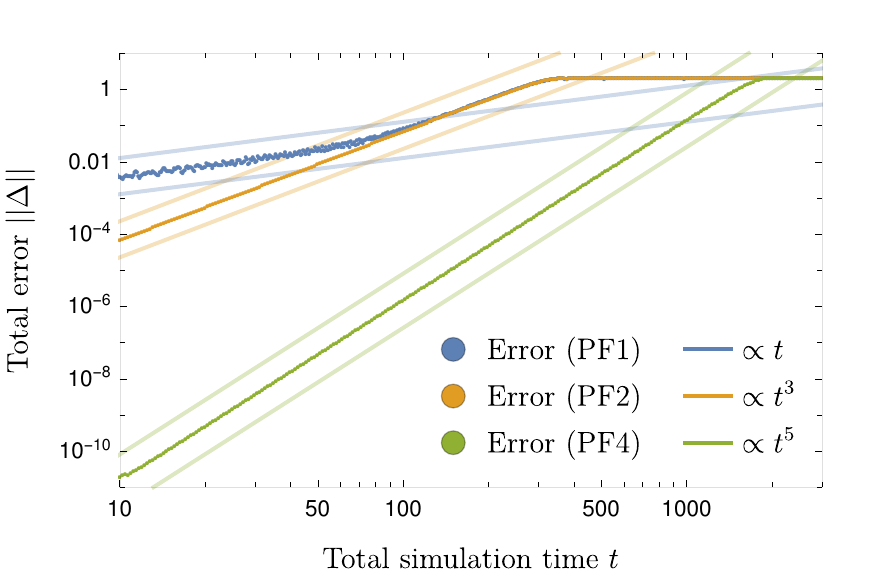}
\caption{The total error of simulating the Heisenberg chain with $n=8$ spins in \cref{EQ_Heisenberg} using PF1 (blue dots), PF2 (orange dots) and PF4 (green dots) is numerically computed at $r=10000$, and variable time $t$ between 10 and 3000.
We also plot functions proportional to $t$ (blue lines), $t^3$ (orange lines), and $t^5$ (green lines) for reference.
}
\label{fig:error_scaling_higher}
\end{figure}

\Section{Gate count}
Given the error bound in \cref{EQ_norm_Delta}, we now count the number of gates for PF1. 
\Cref{EQ_norm_Delta} suggests we should choose
\begin{align}
	r\propto \max\left\{\frac{nt}{\epsilon},\sqrt{\frac{nt^3}{\epsilon}},1\right\},\label{eq:r_two_choices}
\end{align}
so that the total error $\norm\Delta$ is at most $\epsilon$.
First, we assume $nt\geq \epsilon$ and consider two cases, corresponding to $\epsilon t\leq 1$ (small time) and $\epsilon t>1$ (large time).
The former condition implies that
the first term in \cref{eq:r_two_choices} dominates
and therefore we should choose
$
	r =\BigTheta{{nt}/{\epsilon}}.
$
This choice of $r$ together with $\epsilon t\leq 1$ also fulfills the condition $r\norm{\delta}<1/2$ required earlier, as long as we choose a large enough prefactor in $\BigTheta{{nt}/{\epsilon}}$.
Thus, when $\epsilon t\leq 1$, the gate count of PF1 is 
\begin{align}
	\O{rn} = \O{\frac{n^2t}{\epsilon}}.\label{EQ_GC_less}
\end{align}

On the other hand, when $\epsilon t> 1$, we divide the simulation into $m$ stages.
In each stage, we simulate the evolution for time $t/m$ with an error at most $\epsilon/m$ by further dividing the stage into $r$ time segments.
In order to apply the above analysis in each stage, we require $m$ to be large enough so that $\epsilon t/m^2\leq 1$.
Since the resulting gate count $\O{mn^2t/\epsilon}$ increases with $m$, it is optimal to choose $m$ as small as possible, i.e. $m = \ceil{\sqrt{\epsilon t}}$. 
Therefore, the total gate count in this case is 
\begin{align}
	\O{\sqrt{\epsilon t}\frac{n^2 t}{\epsilon}}=\O{\frac{n^2 t^{3/2}}{\epsilon^{1/2}}}.	\label{EQ_GC_more}
\end{align}
Finally, when $nt<\epsilon$, we simply choose $r = \Theta(1)$, giving gate count $\O{1}$. 
Combining the above arguments, we have an upper bound on the total gate count of
\begin{align}
	\max\left\{\O{\frac{n^2t}{\epsilon}},\O{\frac{n^2 t^{3/2}}{\epsilon^{1/2}}},\O{1}\right\},\label{eq:total_gate_count}
\end{align}
which is valid for all times $t$ and is tighter than the previous best estimate in Ref.~\cite{ChildsY19}.

\Section{Discussion \& Outlook}
As mentioned earlier, we assume that the terms of the Hamiltonian can be separated into two parts such that the terms within each part mutually commute. 
Therefore, our results apply to translationally invariant spin chains in one dimension with finite-range interactions and with either open or periodic boundary conditions, as well as disordered spin chains, such as those featuring many-body localization~\cite{Pal2010}. 
Additionally, our analysis also holds for some systems in higher dimensions, such as the transverse field Ising model with either finite-range or long-range interactions, where the two mutually commuting parts of the Hamiltonian are the spin-spin interactions and the field terms.
However, for long-range interactions, the number of interaction terms can scale as $\O{n^2}$ [instead of $\O{n}$ for the finite-range interactions], so the scalings of the error bound and of the gate count as functions of $n$ must be adjusted accordingly. 
Furthermore, our technique can also be used to bound the error in simulating materials where the electronic structure Hamiltonian in the plane wave dual basis~\cite{Babbush2018} is a sum of mutually commuting kinetic energy terms and Coulomb interactions.

However, it is unclear whether our approach generalizes to Hamiltonians that can only be separated into three or more mutually commuting parts, such as those that typically occur in higher dimensions and systems with general long-range interactions, where the simple relation between $\delta$ and $H$ in \cref{EQ_delta_comm} no longer holds in general. 
In addition, although our main focus in this paper is on real-time simulation, it would be interesting to consider the implications of our bound for the error of the product formula in simulating imaginary time evolution, which is relevant for path integral Quantum Monte Carlo algorithms~\cite{Sandvik2010}. 

We also note that while our analysis requires $r \norm\delta<1/2$, our numerical calculation [see \cref{fig:error_scaling}] shows that our error bound agrees well with the empirical scaling even at large values of $t$, where $r \norm \delta \gg 1/2$.
Therefore, we conjecture that the error bound in \cref{EQ_norm_Delta} is valid regardless of whether $\epsilon t$ is less than one. 
If the conjecture holds, \cref{eq:r_two_choices} implies that we should choose $r\propto nt/\epsilon$ and $r\propto \sqrt{nt^3/\epsilon}$ for $\epsilon t\leq n$ and $\epsilon t> n$, respectively (in the limit of large $n$ and $t$). 
The former choice yields the same gate count $\O{n^2t/\epsilon}$ as in \cref{EQ_GC_less},
but the latter choice leads to a gate count of $\O{nr} = \O{\sqrt{n^3t^3/\epsilon}}$, which is tighter than the estimate in \cref{EQ_GC_more}.
Thus, the conjecture would imply that PF1 performs as well as PF2---whose gate count is also $\O{\sqrt{n^3t^3/\epsilon}}$~\cite{ChildsY19}---in the large-time limit.
We consider proving the conjecture a very interesting future direction.

\begin{acknowledgments}
	\Section{Acknowledgments}
	MCT, SKC, and AVG acknowledge funding from DoE ASCR FAR-QC (award No. DE-SC0020312), DoE BES Materials and Chemical Sciences Research for Quantum Information Science program (award No. DE-SC0019449), DoE ASCR Quantum Testbed Pathfinder program (award No. DE-SC0019040), NSF PFCQC program, ARO MURI, AFOSR, ARL CDQI, and NSF PFC at JQI.
	AMC and YS acknowledge funding from ARO MURI, NSF (Grant No. CCF-1813814), and the U.S.\ Department of Energy, Office of Science, Office of Advanced Scientific Computing Research, Quantum Algorithms Teams and Quantum Testbed Pathfinder programs (Award No. DE-SC0019040).
	MCT is supported in part by the NSF Grant No.~NSF PHY-1748958 and the Heising-Simons Foundation.
	SKC also acknowledges the support from the Studying Abroad Scholarship by Ministry of Education in Taiwan (R.O.C.).
\end{acknowledgments}

\bibliographystyle{apsrev4-1}
\bibliography{product-formula}

\begin{thebibliography}{20}%
\makeatletter
\providecommand \@ifxundefined [1]{%
 \@ifx{#1\undefined}
}%
\providecommand \@ifnum [1]{%
 \ifnum #1\expandafter \@firstoftwo
 \else \expandafter \@secondoftwo
 \fi
}%
\providecommand \@ifx [1]{%
 \ifx #1\expandafter \@firstoftwo
 \else \expandafter \@secondoftwo
 \fi
}%
\providecommand \natexlab [1]{#1}%
\providecommand \enquote  [1]{``#1''}%
\providecommand \bibnamefont  [1]{#1}%
\providecommand \bibfnamefont [1]{#1}%
\providecommand \citenamefont [1]{#1}%
\providecommand \href@noop [0]{\@secondoftwo}%
\providecommand \href [0]{\begingroup \@sanitize@url \@href}%
\providecommand \@href[1]{\@@startlink{#1}\@@href}%
\providecommand \@@href[1]{\endgroup#1\@@endlink}%
\providecommand \@sanitize@url [0]{\catcode `\\12\catcode `\$12\catcode
  `\&12\catcode `\#12\catcode `\^12\catcode `\_12\catcode `\%12\relax}%
\providecommand \@@startlink[1]{}%
\providecommand \@@endlink[0]{}%
\providecommand \url  [0]{\begingroup\@sanitize@url \@url }%
\providecommand \@url [1]{\endgroup\@href {#1}{\urlprefix }}%
\providecommand \urlprefix  [0]{URL }%
\providecommand \Eprint [0]{\href }%
\providecommand \doibase [0]{http://dx.doi.org/}%
\providecommand \selectlanguage [0]{\@gobble}%
\providecommand \bibinfo  [0]{\@secondoftwo}%
\providecommand \bibfield  [0]{\@secondoftwo}%
\providecommand \translation [1]{[#1]}%
\providecommand \BibitemOpen [0]{}%
\providecommand \bibitemStop [0]{}%
\providecommand \bibitemNoStop [0]{.\EOS\space}%
\providecommand \EOS [0]{\spacefactor3000\relax}%
\providecommand \BibitemShut  [1]{\csname bibitem#1\endcsname}%
\let\auto@bib@innerbib\@empty
\bibitem [{\citenamefont {Lloyd}(1996)}]{Lloyd1996}%
  \BibitemOpen
  \bibfield  {author} {\bibinfo {author} {\bibfnamefont {S.}~\bibnamefont
  {Lloyd}},\ }\href {\doibase 10.1126/science.273.5278.1073} {\bibfield
  {journal} {\bibinfo  {journal} {Science}\ }\textbf {\bibinfo {volume}
  {273}},\ \bibinfo {pages} {1073} (\bibinfo {year} {1996})}\BibitemShut
  {NoStop}%
\bibitem [{\citenamefont {Suzuki}(1985)}]{Suzuki1985}%
  \BibitemOpen
  \bibfield  {author} {\bibinfo {author} {\bibfnamefont {M.}~\bibnamefont
  {Suzuki}},\ }\href {\doibase 10.1063/1.526596} {\bibfield  {journal}
  {\bibinfo  {journal} {J. Math. Phys.}\ }\textbf {\bibinfo {volume} {26}},\
  \bibinfo {pages} {601} (\bibinfo {year} {1985})}\BibitemShut {NoStop}%
\bibitem [{\citenamefont {Huyghebaert}\ and\ \citenamefont
  {Raedt}(1990)}]{Huyghebaert1990}%
  \BibitemOpen
  \bibfield  {author} {\bibinfo {author} {\bibfnamefont {J.}~\bibnamefont
  {Huyghebaert}}\ and\ \bibinfo {author} {\bibfnamefont {H.~D.}\ \bibnamefont
  {Raedt}},\ }\href {\doibase 10.1088/0305-4470/23/24/019} {\bibfield
  {journal} {\bibinfo  {journal} {J. Phys. A: Math. Gen.}\ }\textbf {\bibinfo
  {volume} {23}},\ \bibinfo {pages} {5777} (\bibinfo {year}
  {1990})}\BibitemShut {NoStop}%
\bibitem [{\citenamefont {Childs}\ and\ \citenamefont {Su}(2019)}]{ChildsY19}%
  \BibitemOpen
  \bibfield  {author} {\bibinfo {author} {\bibfnamefont {A.~M.}\ \bibnamefont
  {Childs}}\ and\ \bibinfo {author} {\bibfnamefont {Y.}~\bibnamefont {Su}},\
  }\href {\doibase 10.1103/PhysRevLett.123.050503} {\bibfield  {journal}
  {\bibinfo  {journal} {Phys. Rev. Lett.}\ }\textbf {\bibinfo {volume} {123}},\
  \bibinfo {pages} {050503} (\bibinfo {year} {2019})}\BibitemShut {NoStop}%
\bibitem [{\citenamefont {Childs}(2004)}]{Childs2004}%
  \BibitemOpen
  \bibfield  {author} {\bibinfo {author} {\bibfnamefont {A.~M.}\ \bibnamefont
  {Childs}},\ }\emph {\bibinfo {title} {Quantum information processing in
  continuous time}},\ \href@noop {} {Ph.D. thesis},\ \bibinfo  {school}
  {Massachusetts Institute of Technology}, \bibinfo {address} {Cambridge, MA}
  (\bibinfo {year} {2004})\BibitemShut {NoStop}%
\bibitem [{\citenamefont {{Berry}}\ \emph {et~al.}(2007)\citenamefont
  {{Berry}}, \citenamefont {{Ahokas}}, \citenamefont {{Cleve}},\ and\
  \citenamefont {{Sanders}}}]{BerryACS07}%
  \BibitemOpen
  \bibfield  {author} {\bibinfo {author} {\bibfnamefont {D.~W.}\ \bibnamefont
  {{Berry}}}, \bibinfo {author} {\bibfnamefont {G.}~\bibnamefont {{Ahokas}}},
  \bibinfo {author} {\bibfnamefont {R.}~\bibnamefont {{Cleve}}}, \ and\
  \bibinfo {author} {\bibfnamefont {B.~C.}\ \bibnamefont {{Sanders}}},\ }\href
  {\doibase 10.1007/s00220-006-0150-x} {\bibfield  {journal} {\bibinfo
  {journal} {Commun. Math. Phys.}\ }\textbf {\bibinfo {volume} {270}},\
  \bibinfo {pages} {359} (\bibinfo {year} {2007})}\BibitemShut {NoStop}%
\bibitem [{\citenamefont {Childs}\ \emph {et~al.}(2018)\citenamefont {Childs},
  \citenamefont {Maslov}, \citenamefont {Nam}, \citenamefont {Ross},\ and\
  \citenamefont {Su}}]{ChildsMNRS2017}%
  \BibitemOpen
  \bibfield  {author} {\bibinfo {author} {\bibfnamefont {A.~M.}\ \bibnamefont
  {Childs}}, \bibinfo {author} {\bibfnamefont {D.}~\bibnamefont {Maslov}},
  \bibinfo {author} {\bibfnamefont {Y.}~\bibnamefont {Nam}}, \bibinfo {author}
  {\bibfnamefont {N.~J.}\ \bibnamefont {Ross}}, \ and\ \bibinfo {author}
  {\bibfnamefont {Y.}~\bibnamefont {Su}},\ }\href {\doibase
  10.1073/pnas.1801723115} {\bibfield  {journal} {\bibinfo  {journal} {Proc.
  Natl. Acad. Sci.}\ }\textbf {\bibinfo {volume} {115}},\ \bibinfo {pages}
  {9456} (\bibinfo {year} {2018})}\BibitemShut {NoStop}%
\bibitem [{\citenamefont {Berry}\ \emph {et~al.}(2015)\citenamefont {Berry},
  \citenamefont {Childs}, \citenamefont {Cleve}, \citenamefont {Kothari},\ and\
  \citenamefont {Somma}}]{BerryCCKS2015}%
  \BibitemOpen
  \bibfield  {author} {\bibinfo {author} {\bibfnamefont {D.~W.}\ \bibnamefont
  {Berry}}, \bibinfo {author} {\bibfnamefont {A.~M.}\ \bibnamefont {Childs}},
  \bibinfo {author} {\bibfnamefont {R.}~\bibnamefont {Cleve}}, \bibinfo
  {author} {\bibfnamefont {R.}~\bibnamefont {Kothari}}, \ and\ \bibinfo
  {author} {\bibfnamefont {R.~D.}\ \bibnamefont {Somma}},\ }\href {\doibase
  10.1103/PhysRevLett.114.090502} {\bibfield  {journal} {\bibinfo  {journal}
  {Phys. Rev. Lett.}\ }\textbf {\bibinfo {volume} {114}},\ \bibinfo {pages}
  {090502} (\bibinfo {year} {2015})}\BibitemShut {NoStop}%
\bibitem [{\citenamefont {{Low}}\ \emph {et~al.}(2019)\citenamefont {{Low}},
  \citenamefont {{Kliuchnikov}},\ and\ \citenamefont {{Wiebe}}}]{Low19}%
  \BibitemOpen
  \bibfield  {author} {\bibinfo {author} {\bibfnamefont {G.~H.}\ \bibnamefont
  {{Low}}}, \bibinfo {author} {\bibfnamefont {V.}~\bibnamefont
  {{Kliuchnikov}}}, \ and\ \bibinfo {author} {\bibfnamefont {N.}~\bibnamefont
  {{Wiebe}}},\ }\href@noop {} {\  (\bibinfo {year} {2019})},\ \Eprint
  {http://arxiv.org/abs/1907.11679} {arXiv:1907.11679} \BibitemShut {NoStop}%
\bibitem [{\citenamefont {Low}\ and\ \citenamefont {Chuang}(2017)}]{LowC2017}%
  \BibitemOpen
  \bibfield  {author} {\bibinfo {author} {\bibfnamefont {G.~H.}\ \bibnamefont
  {Low}}\ and\ \bibinfo {author} {\bibfnamefont {I.~L.}\ \bibnamefont
  {Chuang}},\ }\href {\doibase 10.1103/PhysRevLett.118.010501} {\bibfield
  {journal} {\bibinfo  {journal} {Phys. Rev. Lett.}\ }\textbf {\bibinfo
  {volume} {118}},\ \bibinfo {pages} {010501} (\bibinfo {year}
  {2017})}\BibitemShut {NoStop}%
\bibitem [{\citenamefont {{Haah}}\ \emph {et~al.}(2018)\citenamefont {{Haah}},
  \citenamefont {{Hastings}}, \citenamefont {{Kothari}},\ and\ \citenamefont
  {{Hao Low}}}]{Haah}%
  \BibitemOpen
  \bibfield  {author} {\bibinfo {author} {\bibfnamefont {J.}~\bibnamefont
  {{Haah}}}, \bibinfo {author} {\bibfnamefont {M.~B.}\ \bibnamefont
  {{Hastings}}}, \bibinfo {author} {\bibfnamefont {R.}~\bibnamefont
  {{Kothari}}}, \ and\ \bibinfo {author} {\bibfnamefont {G.}~\bibnamefont {{Hao
  Low}}},\ }\href@noop {} {\  (\bibinfo {year} {2018})},\ \Eprint
  {http://arxiv.org/abs/1801.03922} {arXiv:1801.03922} \BibitemShut {NoStop}%
\bibitem [{\citenamefont {Tran}\ \emph {et~al.}(2019)\citenamefont {Tran},
  \citenamefont {Guo}, \citenamefont {Su}, \citenamefont {Garrison},
  \citenamefont {Eldredge}, \citenamefont {Foss-Feig}, \citenamefont {Childs},\
  and\ \citenamefont {Gorshkov}}]{Tran2018}%
  \BibitemOpen
  \bibfield  {author} {\bibinfo {author} {\bibfnamefont {M.~C.}\ \bibnamefont
  {Tran}}, \bibinfo {author} {\bibfnamefont {A.~Y.}\ \bibnamefont {Guo}},
  \bibinfo {author} {\bibfnamefont {Y.}~\bibnamefont {Su}}, \bibinfo {author}
  {\bibfnamefont {J.~R.}\ \bibnamefont {Garrison}}, \bibinfo {author}
  {\bibfnamefont {Z.}~\bibnamefont {Eldredge}}, \bibinfo {author}
  {\bibfnamefont {M.}~\bibnamefont {Foss-Feig}}, \bibinfo {author}
  {\bibfnamefont {A.~M.}\ \bibnamefont {Childs}}, \ and\ \bibinfo {author}
  {\bibfnamefont {A.~V.}\ \bibnamefont {Gorshkov}},\ }\href {\doibase
  10.1103/PhysRevX.9.031006} {\bibfield  {journal} {\bibinfo  {journal} {Phys.
  Rev. X}\ }\textbf {\bibinfo {volume} {9}},\ \bibinfo {pages} {031006}
  (\bibinfo {year} {2019})}\BibitemShut {NoStop}%
\bibitem [{Note1()}]{Note1}%
  \BibitemOpen
  \bibinfo {note} {Refs.~\cite {Suzuki1985,Huyghebaert1990,ChildsY19} took into
  account the commutativity between some interaction terms in the Hamiltonian
  of a nearest-neighbor interacting system. Without this commutativity, the
  gate count would be $O\left (n^3t^2\right )$~\cite
  {Childs2004,BerryACS07}}\BibitemShut {NoStop}%
\bibitem [{\citenamefont {{Heyl}}\ \emph {et~al.}(2019)\citenamefont {{Heyl}},
  \citenamefont {{Hauke}},\ and\ \citenamefont {{Zoller}}}]{Heyl18}%
  \BibitemOpen
  \bibfield  {author} {\bibinfo {author} {\bibfnamefont {M.}~\bibnamefont
  {{Heyl}}}, \bibinfo {author} {\bibfnamefont {P.}~\bibnamefont {{Hauke}}}, \
  and\ \bibinfo {author} {\bibfnamefont {P.}~\bibnamefont {{Zoller}}},\ }\href
  {\doibase 10.1126/sciadv.aau8342} {\bibfield  {journal} {\bibinfo  {journal}
  {Sci. Adv.}\ }\textbf {\bibinfo {volume} {5}},\ \bibinfo {pages} {eaau8342}
  (\bibinfo {year} {2019})}\BibitemShut {NoStop}%
\bibitem [{Note2()}]{Note2}%
  \BibitemOpen
  \bibinfo {note} {For interactions of maximum range $R$, we can group
  $\protect \ceil { (R+1)/2 }$ consecutive sites into distinct blocks such that
  the Hamiltonian consists of only interactions between nearest-neighbor
  blocks. The error analysis for using the first-order product formula to
  simulate such a system would follow from our analysis for simulating
  nearest-neighbor interactions. Note, however, that we assume that the exact
  simulation of the evolution of each constant-size block requires only a
  constant amount of elementary gates.}\BibitemShut {Stop}%
\bibitem [{\citenamefont {{Berry}}\ \emph {et~al.}(2015)\citenamefont
  {{Berry}}, \citenamefont {{Childs}},\ and\ \citenamefont
  {{Kothari}}}]{BerryCK15}%
  \BibitemOpen
  \bibfield  {author} {\bibinfo {author} {\bibfnamefont {D.~W.}\ \bibnamefont
  {{Berry}}}, \bibinfo {author} {\bibfnamefont {A.~M.}\ \bibnamefont
  {{Childs}}}, \ and\ \bibinfo {author} {\bibfnamefont {R.}~\bibnamefont
  {{Kothari}}},\ }in\ \href {\doibase 10.1109/FOCS.2015.54} {\emph {\bibinfo
  {booktitle} {2015 IEEE 56th Annual Symposium on Foundations of Computer
  Science}}}\ (\bibinfo {year} {2015})\ pp.\ \bibinfo {pages}
  {792--809}\BibitemShut {NoStop}%
\bibitem [{SM()}]{SM}%
  \BibitemOpen
  \href@noop {} {}\bibinfo {note} {See Supplemental Material for detailed
  proofs of \cref{TH_delta_k,LEM_sum_delta,LEM_Delta_k_norm}.}\BibitemShut
  {Stop}%
\bibitem [{\citenamefont {Pal}\ and\ \citenamefont {Huse}(2010)}]{Pal2010}%
  \BibitemOpen
  \bibfield  {author} {\bibinfo {author} {\bibfnamefont {A.}~\bibnamefont
  {Pal}}\ and\ \bibinfo {author} {\bibfnamefont {D.~A.}\ \bibnamefont {Huse}},\
  }\href {\doibase 10.1103/PhysRevB.82.174411} {\bibfield  {journal} {\bibinfo
  {journal} {Phys. Rev. B}\ }\textbf {\bibinfo {volume} {82}},\ \bibinfo
  {pages} {174411} (\bibinfo {year} {2010})}\BibitemShut {NoStop}%
\bibitem [{\citenamefont {Babbush}\ \emph {et~al.}(2018)\citenamefont
  {Babbush}, \citenamefont {Wiebe}, \citenamefont {McClean}, \citenamefont
  {McClain}, \citenamefont {Neven},\ and\ \citenamefont {Chan}}]{Babbush2018}%
  \BibitemOpen
  \bibfield  {author} {\bibinfo {author} {\bibfnamefont {R.}~\bibnamefont
  {Babbush}}, \bibinfo {author} {\bibfnamefont {N.}~\bibnamefont {Wiebe}},
  \bibinfo {author} {\bibfnamefont {J.}~\bibnamefont {McClean}}, \bibinfo
  {author} {\bibfnamefont {J.}~\bibnamefont {McClain}}, \bibinfo {author}
  {\bibfnamefont {H.}~\bibnamefont {Neven}}, \ and\ \bibinfo {author}
  {\bibfnamefont {G.~K.-L.}\ \bibnamefont {Chan}},\ }\href {\doibase
  10.1103/PhysRevX.8.011044} {\bibfield  {journal} {\bibinfo  {journal} {Phys.
  Rev. X}\ }\textbf {\bibinfo {volume} {8}},\ \bibinfo {pages} {011044}
  (\bibinfo {year} {2018})}\BibitemShut {NoStop}%
\bibitem [{\citenamefont {{Sandvik}}(2010)}]{Sandvik2010}%
  \BibitemOpen
  \bibfield  {author} {\bibinfo {author} {\bibfnamefont {A.~W.}\ \bibnamefont
  {{Sandvik}}},\ }in\ \href {\doibase 10.1063/1.3518900} {\emph {\bibinfo
  {booktitle} {American Institute of Physics Conference Series}}},\ \bibinfo
  {series} {American Institute of Physics Conference Series}, Vol.\ \bibinfo
  {volume} {1297},\ \bibinfo {editor} {edited by\ \bibinfo {editor}
  {\bibfnamefont {A.}~\bibnamefont {{Avella}}}\ and\ \bibinfo {editor}
  {\bibfnamefont {F.}~\bibnamefont {{Mancini}}}}\ (\bibinfo {year} {2010})\
  pp.\ \bibinfo {pages} {135--338},\ \Eprint {http://arxiv.org/abs/1101.3281}
  {arXiv:1101.3281 [cond-mat.str-el]} \BibitemShut {NoStop}%
\end{thebibliography}%


\begin{thebibliography}{0}%
\makeatletter
\providecommand \@ifxundefined [1]{%
 \@ifx{#1\undefined}
}%
\providecommand \@ifnum [1]{%
 \ifnum #1\expandafter \@firstoftwo
 \else \expandafter \@secondoftwo
 \fi
}%
\providecommand \@ifx [1]{%
 \ifx #1\expandafter \@firstoftwo
 \else \expandafter \@secondoftwo
 \fi
}%
\providecommand \natexlab [1]{#1}%
\providecommand \enquote  [1]{``#1''}%
\providecommand \bibnamefont  [1]{#1}%
\providecommand \bibfnamefont [1]{#1}%
\providecommand \citenamefont [1]{#1}%
\providecommand \href@noop [0]{\@secondoftwo}%
\providecommand \href [0]{\begingroup \@sanitize@url \@href}%
\providecommand \@href[1]{\@@startlink{#1}\@@href}%
\providecommand \@@href[1]{\endgroup#1\@@endlink}%
\providecommand \@sanitize@url [0]{\catcode `\\12\catcode `\$12\catcode
  `\&12\catcode `\#12\catcode `\^12\catcode `\_12\catcode `\%12\relax}%
\providecommand \@@startlink[1]{}%
\providecommand \@@endlink[0]{}%
\providecommand \url  [0]{\begingroup\@sanitize@url \@url }%
\providecommand \@url [1]{\endgroup\@href {#1}{\urlprefix }}%
\providecommand \urlprefix  [0]{URL }%
\providecommand \Eprint [0]{\href }%
\providecommand \doibase [0]{http://dx.doi.org/}%
\providecommand \selectlanguage [0]{\@gobble}%
\providecommand \bibinfo  [0]{\@secondoftwo}%
\providecommand \bibfield  [0]{\@secondoftwo}%
\providecommand \translation [1]{[#1]}%
\providecommand \BibitemOpen [0]{}%
\providecommand \bibitemStop [0]{}%
\providecommand \bibitemNoStop [0]{.\EOS\space}%
\providecommand \EOS [0]{\spacefactor3000\relax}%
\providecommand \BibitemShut  [1]{\csname bibitem#1\endcsname}%
\let\auto@bib@innerbib\@empty
\end{thebibliography}%

\end{document}


\title{Supplemental Material for: ``Destructive Error Interference in Product-Formula Lattice Simulation''}

\author{Minh~C.~Tran}
\affiliation{Joint Center for Quantum Information and Computer Science, NIST/University of Maryland, College Park, Maryland 20742, USA}
\affiliation{Joint Quantum Institute, NIST/University of Maryland, College Park, Maryland 20742, USA}
\affiliation{Kavli Institute for Theoretical Physics, University of California, Santa Barbara, California 93106, USA}
\author{Su-Kuan Chu}
\affiliation{Joint Center for Quantum Information and Computer Science, NIST/University of Maryland, College Park, Maryland 20742, USA}
\affiliation{Joint Quantum Institute, NIST/University of Maryland, College Park, Maryland 20742, USA}
\author{Yuan~Su}
\affiliation{Joint Center for Quantum Information and Computer Science, NIST/University of Maryland, College Park, Maryland 20742, USA}
\affiliation{Department of Computer Science, University of Maryland, College Park, Maryland 20742, USA}
\affiliation{Institute for Advanced Computer Studies, University of Maryland, College Park, Maryland 20742, USA}
\author{Andrew~M.~Childs}
\affiliation{Joint Center for Quantum Information and Computer Science, NIST/University of Maryland, College Park, Maryland 20742, USA}
\affiliation{Department of Computer Science, University of Maryland, College Park, Maryland 20742, USA}
\affiliation{Institute for Advanced Computer Studies, University of Maryland, College Park, Maryland 20742, USA}
\author{Alexey~V.~Gorshkov}
\affiliation{Joint Center for Quantum Information and Computer Science, NIST/University of Maryland, College Park, Maryland 20742, USA}
\affiliation{Joint Quantum Institute, NIST/University of Maryland, College Park, Maryland 20742, USA}
\affiliation{Kavli Institute for Theoretical Physics, University of California, Santa Barbara, California 93106, USA}\date{\today}
\maketitle

The Supplemental Material provides more mathematical details for the derivations of the error bound in the paper.
Specifically, \cref{APP_delta_norm_proof} explains how we write the $k$-th order error $\delta_k$ into a commutator. 
\Cref{APP_sum_delta_proof} provides an upper bound for a sum of different evolutions of $\delta$. 
Finally, in \cref{APP_Delta_k_norm_proof}, we show how we bound the norm of $\Delta_k$ in \cref{LEM_Delta_k_norm}.

\section{Structure of $\delta_k$}
\label{APP_delta_norm_proof}
In this section, we present the proof of \Cref{TH_delta_k}, which says that we can write $\delta_k$ into a sum of a commutator and an operator of higher order. 
First, we need the following recursive relation between the $\delta_k$ operators.
\begin{lemma}
\label{lemma:deltarecursion}
For $k\geq2$, we have the following recursive relation:
\begin{align}
	\delta_{k+1}=H_{\o}\delta_{k}+\delta_{k}H_{\e}-[H^{k},H_{\e}].
\end{align}
	
\end{lemma}

\begin{proof}[Proof]
We prove the lemma by expanding both $U_{t/r}$ and $U_{t/r}^{(\o)}U_{t/r}^{(\e)}$ in orders of $t/r$:
\begin{align}
	&U_{t/r}^{(\o)}U_{t/r}^{(\e)} = e^{-iH_{\o}t/r}e^{-iH_{\e}t/r}=\sum_{k=0}^{\infty}\frac{1}{k!}A_{k}\left(\frac{-it}{r}\right)^{k},\\
	&U_{t/r} = e^{-iHt/r}=\sum_{k=0}^{\infty}\frac{1}{k!}B_{k}\left(\frac{-it}{r}\right)^{k},
\end{align}
where
\begin{align}
	A_{k}\coloneqq\sum_{j=0}^{k}\binom{k}{j} H_{\o}^{j}H_{\e}^{k-j},\:\:\:\:\:
	B_{k}\coloneqq H^{k}=(H_{\o}+H_{\e})^{k}.
\end{align}
With these notations, we have the relation $\delta_{k}=B_{k}-A_{k}$. 
It is also straightforward to verify the recursive relations for $A_k$ and $B_k$:
\begin{align}
	A_{k+1}&=H_{\o}A_{k}+A_{k}H_{\e},\\
	B_{k+1} & =H^{k+1}=HB_k=(H_{\o}+H_{\e})(A_{k}+\delta_{k})\nonumber \\
	& =H_{\o}A_{k}+H_{\o}\delta_{k}+B_{k}H_{\e}-[B_{k},H_{\e}]\nonumber \\
	& =H_{\o}A_{k}+H_{\o}\delta_{k}+(A_{k}+\delta_{k})H_{\e}-[B_{k},H_{\e}]\nonumber \\
	& =(H_{\o}A_{k}+A_{k}H_{\e})+H_{\o}\delta_{k}+\delta_{k}H_{\e}-\left[H^{k},H_{\e}\right]\nonumber \\
	& =A_{k+1}+H_{\o}\delta_{k}+\delta_{k}H_{\e}-\left[H^{k},H_{\e}\right].
\end{align}
By definition, we have	
	\begin{align}
	\delta_{k+1}=B_{k+1}-A_{k+1}=H_{\o}\delta_{k}+\delta_{k}H_{\e}-\left[H^{k},H_{\e}\right].
	\end{align}
Therefore, the lemma follows.
\end{proof}

We now construct the operators $S_k,V_k$ in \Cref{TH_delta_k} inductively on $k$. 
For $k = 2$, we have $\delta_2 = \comm{H,H_\e}$.
Thus \Cref{TH_delta_k} is true for $k = 2$ with $S_2 = H_\e$ and $V_2 = 0$.
Assume that \Cref{TH_delta_k} is true up to $k$, i.e. there exist $S_k,V_k$ such that $\delta_k = \comm{H,S_k}+V_k$, we shall prove that it is also true for $k+1$.
Using \Cref{lemma:deltarecursion}, we have
\begin{align}
\delta_{k+1} & =H_{\o}\delta_{k}+\delta_{k}H_{\e}-[H^{k},H_{\e}]\nonumber \\
& =[H_{\o},\delta_{k}]+\delta_{k}H-[H^{k},H_{\e}]\nonumber \\
& =[H_{\o},[H,S_{k}]+V_{k}]+V_{k}H+[H,S_{k}]H-[H^{k},H_{\e}].
\end{align}

We use the following commutator identities: 
\begin{align}
	[H,S_{k}]H&=[H,S_{k}H],\\
	[H^{k},H_{\e}] & =[H,\sum_{j=0}^{k-1}H^{k-1-j}H_{\e}H^{j}].
\end{align}
With some trivial manipulations, we can write $\delta_{k+1} = [H,S_{k+1}]+V_{k+1}$, where 
\begin{align}
	&S_{k+1} = S_{k}H-\sum_{j=0}^{k-1}H^{k-1-j}H_{\e}H^{j},\label{EQ_Sk_recursive}\\
	&V_{k+1} = [H_{\o},[H,S_{k}]]+H_{\o}V_{k}+V_{k}H_{\e}.\label{EQ_Vk_recursive}
\end{align}
Finally, we show that the operators $S_k,V_k$ constructed using the above recursive relations satisfy the norm bounds in \cref{EQ_Vk_norm,EQ_Sk_norm,EQ_HSk_norm}. 
We need the following lemma about the structure of $S_k,V_k$.
\begin{lemma}
\label{lemma:boundSV}
For integer $k\geq2$, the operators $S_k,V_k$ constructed from \cref{EQ_Sk_recursive,EQ_Vk_recursive} can be written as
\begin{align}
	V_{k}=\sum_{i=1}^{n_{k}}v_{k,i},\quad n_{k}\leq \xi e^{k-2} n^{k-2},\label{EQ_Vk_sum}
\end{align}
\begin{align}
	S_{k}=\sum_{i=1}^{m_{k}}s_{k,i},\quad m_{k}\leq\frac{k(k-1)}{2}n^{k-1}, \label{EQ_Sk_sum}
\end{align}
where $\xi$ is a constant, $v_{k,i},s_{k,i}$ are operators supported on at most $2(k-1)$ sites and 
\begin{align}
	\norm{s_{k,i}}\leq 1,\quad
	\norm{v_{k,i}}\leq 1,
\end{align}
for all $i$.
\end{lemma}

\begin{proof}[Proof]
Denote by $\supp{X}$ the support size of an operator $X$, i.e. the number of sites $X$ acts nontrivially on. 
We say that the number of terms in $V_k$ is $x$ if there exists a decomposition $V_k = \sum_{j=1}^x v_j$ such that $\norm{v_j}\leq 1$ for all $j$. 
For $k = 2$, the lemma is true by definition.
Assume that the lemma is true up to some $k\geq 2$, we shall prove that it holds for $k+1$.

First, we argue for the bounds on the number of terms $m_{k+1},n_{k+1}$ in $S_{k+1},V_{k+1}$ respectively.
Since there are $m_k$ terms in $S_k$, using \cref{EQ_Sk_recursive}, it is straightforward to bound $m_{k+1}$---the number of terms in $S_{k+1}$:
\begin{align}
	m_{k+1}\leq m_kn + k n^{k} \leq \frac{k(k-1)}{2}n^k + kn^k = \frac{k(k+1)}{2}n^{k}.
\end{align}
To bound $n_{k+1}$, the number of terms in $V_{k+1}$, we use \cref{EQ_Vk_recursive} and note that $s_{k,i}$ can non-commute with at most $2\supp{s_{k,i}} = 4(k-1)$ terms from $H$.
Therefore, the number of terms in $\comm{H,S_k}$ is at most $4(k-1)m_k$.
Each of these terms has its support size increased by at most one (to $2k-1$) compared to the terms of $S_k$.
Repeating the argument for $\comm{H_1,\comm{H,S_k}}$, the number of terms in $V_{k+1}$ can be bounded as follow:
\begin{align}
	n_{k+1} &\leq 2(2k-1)4(k-1)m_k + n n_k\\ 
	&\leq 8k^4 n^{k-1} + \xi e^{k-2} n^{k-1}\\
	&< 2\xi e^{k-2} n^{k-1} 
	< \xi e^{k-1} n^{k-1},
\end{align} 
where $\xi = \frac{2048}{e^2(e-1)}$ and we have used the fact that $8k^4 + \xi e^{k-2} < \xi e^{k-1}$ for all $k\geq 2$.
Therefore, the number of terms $n_{k+1},m_{k+1}$ are bounded according to \cref{EQ_Vk_sum,EQ_Sk_sum}.

It is also apparent from this construction that each iteration in \cref{EQ_Sk_recursive,EQ_Vk_recursive} increases the support size of the constituent terms in $S_k,V_k$ by at most 2. 
Therefore, \Cref{lemma:boundSV} follows.
\end{proof}
With Lemma \ref{lemma:boundSV}, it is straightforward to show that the norms of $V_k,S_k,\comm{H,S_k}$ are upper bounded by the their number of terms:
\begin{align}
	\norm{V_k} &\leq n_k  = \O{e^{k-2} n^{k-2}}\\
	\norm{S_k} &\leq m_k  = \O{k^2 n^{k-1}},\\
	\norm{\comm{H,S_k}} &\leq 4(k-1) m_k  = \O{k^3 n^{k-1}}.
\end{align} 
These bounds complete the proof of \Cref{TH_delta_k}.

\section{Sum of evolutions of \texorpdfstring{$\delta$}{delta}}
\label{APP_sum_delta_proof}
In this section, we present the proof of \cref{LEM_sum_delta}, which provides an upper bound for the sum of evolution of an operator with different times.
\begin{proof}
We denote by $\tau \coloneqq t/r$ and 
\begin{align}
	\Sigma_a(X) \coloneqq \sum_{j=0}^{a-1} U_{j\tau} \comm{H,X} U_{j\tau}^\dag \tau,
\end{align}
where 
$X$ is an arbitrary time-independent operator, $a$ is a positive integer, and $U_t = \exp(-iHt)$ as before. 

First, we need to turn the sum $\Sigma_a(X)$ into a sum of several integrals using the following lemma.
\begin{lemma}
\label{LEM_Sigma_Sum}
Define 
\begin{align}
F[X]&\coloneqq -\frac{1}{\tau}\int_0^{\tau} ds \int_0^{s} dv U_{v}\comm{H,X}U_{v}^\dag,\\
I_t(X)&\coloneqq \int_0^t U_{s} \comm{H,X} U_{s}^\dag ds.
\end{align}
For all $\tau$ such that $n\tau<1$, where $n$ is the number of sites in the system, we have 
\begin{align}
\Sigma_a(X) = \sum_{k=0}^\infty I_{a\tau}(F^{\circ k}[X])), \label{EQ_Sigma_sum}
\end{align}
where $F^{\circ k}$ the $k$-th iterate of a function $F$, i.e. the composition $F^{\circ k}[X] = F[F[\dots F[X]\dots]]$, with $F^{\circ 0}$ being the identity function.
\end{lemma}
\begin{proof}
To prove the claim, we note that 
\begin{align}
	I_{a\tau}(X) = \int_0^{a\tau} U_s\comm{H,X}U_s^\dag ds 
	= \sum_{j=0}^{a-1} \int_{j\tau}^{(j+1)\tau} U_s\comm{H,X}U_s^\dag ds
	= \sum_{j=0}^{a-1} U_{j\tau }\left(\int_{0}^{\tau} U_s\comm{H,X}U_s^\dag ds \right)U^\dag_{j\tau }.
\end{align}
Therefore, we have
\begin{align}
	\Sigma_a(X) - I_{a\tau}(X) 
	&= \sum_{j=0}^{a-1}U_{j\tau}\left(\comm{H,X}\tau - \int_0^{\tau} U_s \comm{H,X}U_s^\dag ds\right)U_{j\tau}^\dag \nonumber\\
	&= \sum_{j=0}^{a-1}U_{j\tau}\int_0^{\tau} ds\left(\comm{H,X} - U_s \comm{H,X}U_s^\dag \right)U_{j\tau}^\dag\nonumber\\ 
	&= \sum_{j=0}^{a-1}U_{j\tau}\int_0^{\tau} ds\int_s^0 dv U_v \comm{H,\comm{H,X}}U_v^\dag U_{j\tau}^\dag\nonumber\\
	&= \sum_{j=0}^{a-1}U_{j\tau} \comm{H,\frac{1}{\tau}\int_0^{\tau} ds\int_s^0 dv U_v\comm{H,X}U_v^\dag} U_{j\tau}^\dag  \tau\nonumber\\
	&= \Sigma_a(F[X])).\label{EQ_Sigma_recursive}
\end{align}
To get the second last line, we use the fact that $H$ and $U_t$ commute in order to move the integral inside the commutator. 
Repeated applications of this recursive relation yields \cref{EQ_Sigma_sum}.
The condition $n\tau<1$ ensures that the sum in \cref{EQ_Sigma_sum} converges (See \cref{LEM_Sigma}).
\end{proof}

\cref{LEM_Sigma} below is a consequence of \cref{LEM_Sigma_Sum}.
\begin{lemma}
	\label{LEM_Sigma}
	If $X$ is time-independent and $\mu\coloneqq\frac{n t}{r}<1$,  $\norm{\Sigma_a(X)} \leq \frac{2}{1-\mu}\norm{X}$. 
\end{lemma}
\begin{proof}
To prove \cref{LEM_Sigma}, we note that 
\begin{align}
	\norm{F[X]}\leq \tau \norm{H} \norm{X}\leq \mu\norm{X}.
\end{align}
Therefore, $\norm{F^{\circ k}[X]}\leq \mu^k\norm{X}$.
Note also that for the time-independent $X$,
\begin{align}
	I_{a\tau}(X)&= \int_0^{a\tau} U_{s} \comm{H,X} U_{s}^\dag ds = U_{a\tau} X U_{a\tau}^\dag - X,
\end{align}
and therefore $\norm{I_{a\tau}(X)} \leq 2\norm{X}$. 
Using \cref{LEM_Sigma_Sum}, we have
\begin{align}
\norm{\Sigma_a(X)} &\leq \sum_{k=0}^\infty \norm{I_{a\tau}(F^{\circ k}[X])} 
\leq 2\sum_{k=0}^\infty \norm{F^{\circ k}[X]}\nonumber\\
&\leq 2\norm X\sum_{k=0}^\infty \mu^k 
=\frac{2}{1-\mu}\norm X\nonumber\\
&=\O{\norm X},
\end{align}
where we have assumed $\mu = \frac{nt}{r}<1$ so that the sum converges. 
Therefore, the lemma follows.
\end{proof}

To prove the \cref{LEM_sum_delta}, we write $\delta = \comm{H,S} + V$ with $S,V$ bounded by \cref{EQ_norm_SV}.
We then use \Cref{LEM_Sigma} with $X=S$ and the triangle inequality to get
\begin{align}
	\Norm{\sum_{j=0}^{a-1} U_{j\tau} \delta \: U_{j\tau}^\dag} 
	&\leq\Norm{\frac{1}{\tau}\Sigma_a(S) }+\Norm{ \sum_{j=0}^{a-1} U_{j\tau} V \: U_{j\tau}^\dag}\\ 
	&= \O{\frac{1}{\tau}\norm{S}} + \O{a\norm{V}} \\
	&=  \O{\frac{nt}{r}}+\O{a\frac{nt^3}{r^3}}.
\end{align}
Thus, the lemma follows.	
\end{proof}
\section{Upper bound on \texorpdfstring{$\Delta_k$}{Delta}}
\label{APP_Delta_k_norm_proof}
In this section, we show how we bound the norms of $\Delta_k$ in \cref{LEM_Delta_k_norm}.
For that, we use \cref{LEM_sum_delta} together with the bound on $\norm{\delta}$ [\cref{EQ_norm_delta}]:
\begin{widetext}
\begin{align}
\left\Vert \Delta_{k}\right\Vert  & =\left\Vert \sum_{i_{1}=0}^{r-k}\sum_{i_{2}=0}^{r-k-i_{1}}\sum_{i_{3}=0}^{r-k-i_{1}-i_{2}}\cdots\sum_{i_{k}=0}^{r-k-i_{1}-i_{2}-\cdots-i_{k}}\underbrace{U_{t/r}^{i_{1}}\delta U_{t/r}^{i_{2}}\delta U_{t/r}^{i_{3}}\delta\cdots}_{\delta\,\mathrm{appears\,}k\,\mathrm{times}}U_{t/r}^{r-k-i_{1}-i_{2}-\cdots-i_{k}}\right\Vert \nonumber \\
& \leq\sum_{i_{1}=0}^{r-k}\sum_{i_{2}=0}^{r-k-i_{1}}\sum_{i_{3}=0}^{r-k-i_{1}-i_{2}}\cdots\left\Vert \delta\right\Vert ^{k-1}\left\Vert \sum_{i_{k}=0}^{r-k-i_{1}-i_{2}-\cdots-i_{k}}U_{t/r}^{i_{k}}\delta U_{t/r}^{-i_{k}}\right\Vert \nonumber \\
& \leq r^{k-1}\left\Vert \delta\right\Vert ^{k-1}\O{\frac{n t}{r}+\frac{nt^3}{r^2}}.
\end{align}
\end{widetext}
Thus, \cref{LEM_Delta_k_norm} follows.

\makeatletter
\renewcommand\@biblabel[1]{[S#1]}
\makeatother
\bibliographystyle{apsrev4-1}
\bibliography{product-formula}

\makeatletter\@input{crossref.tex}\makeatother